\documentclass[a4paper,1p,authoryear]{elsarticle}

\makeatletter
\def\ps@pprintTitle{%
 \let\@oddhead\@empty
 \let\@evenhead\@empty
 \def\@oddfoot{\centerline{\thepage}}%
 \let\@evenfoot\@oddfoot}
\makeatother

\usepackage{verbatim}
\usepackage{graphicx}
\usepackage{url}
\usepackage{listings}
\usepackage{subfigure}

\usepackage{amsmath,amssymb,amsthm,mathrsfs,amsfonts}
\usepackage{amssymb,latexsym}

{\theoremstyle{plain}
  \newtheorem{definition}{Definition}

  \newtheorem{lemma}{Lemma}[section]

  \newtheorem{theorem}{Theorem}

  \newtheorem{corollary}[definition]{Corollary}

}
{\theoremstyle{definition}

}

\newtheorem{remark}{Remark}

\usepackage{booktabs}

\newcommand{\abs}[1]{\lvert#1\rvert}

\DeclareMathOperator*{\Ex}{E}

\DeclareMathOperator{\Cov}{Cov}

\DeclareMathOperator{\sign}{sign}

\newcommand{\spower}[2]{\sign{(#1)}\abs{#1}^{#2}}
\newcommand{\brspower}[2]{{#1}^{<#2>}}

\newcommand{\IID}{\text{IID}}

\newcommand{\tssymb}{X}
\newcommand{\tsetasymb}{\eta}

\newcommand{\tsmain}[1]{\tssymb_{#1}}        

\newcommand{\tseta}[1]{\tsetasymb_{#1}}
\newcommand{\tsvolat}[1]{\sigma_{#1}}

\newcommand{\ptsmain}[1]{\{\tsmain{#1}\}}

\newcommand{\ptseta }[1]{\{\tseta{#1} \}}

\newcommand{\Ptsmain}[1]{\{\tsmain{#1}\}_{#1 \in \mathbb{Z}}}

\newcommand{\Ptseta }[1]{\{\tseta{#1} \}_{#1 \in \mathbb{Z}}}

\newcommand{\tsspsymb}{Y}                
\newcommand{\tssp}[1]{\tsspsymb_{#1}}
\newcommand{\ptssp}[1]{\{\tssp{#1}\}}
\newcommand{\Ptssp}[1]{\{\tssp{#1}\}_{#1\in \mathbb{Z}}}

\newcommand{\tsspeta}[1]{\xi_{#1}}
\newcommand{\ptsspeta}[1]{\{\tsspeta{#1}\}}

\newcommand{\tsmainfield}[1]{\mathscr{F}_{#1}}
\newcommand{\tsmainfieldX}[2]{\mathscr{F}_{#1}^{#2}}

\newcommand{\acvf}[1]{\gamma_{#1}}
\newcommand{\acf}[1]{\rho_{#1}}
\newcommand{\acvfhat}[1]{\hat{\gamma}_{#1}}
\newcommand{\acfhat}[1]{\hat{\rho}_{#1}}

\newcommand{\boldacf}{\boldsymbol \rho}

\newcommand{\boldacfhat}{\hat{\boldsymbol \rho}}

\usepackage{setspace}   
\begin{document}


\begin{frontmatter}
\title{A signed power transformation with application to white noise testing}

\author[1]{Georgi~N.~Boshnakov\corref{cor1}}                                  
\ead{georgi.boshnakov@manchester.ac.uk}

\author[1]{Davide Ravagli}
\ead{davide.ravagli@manchester.ac.uk}

\affiliation[1]{%
  organization = {Department of Mathematics, The University of Manchester},
  addressline = {Oxford Road},
  city = {Manchester},
  postcode = {M13 9PL},
  country = {UK},
}

\cortext[cor1]{Corresponding author}
  
\begin{abstract}
  

  We show that signed power transforms of some ARCH-type processes give ARCH-type processes.
  The class of ARCH-type models for which this property holds contains many common ARCH
  and GARCH models. The results can be useful in testing for white noise when fourth moments
  don't exist and detecting white noise that is not ARCH-type.

\end{abstract}

\begin{keyword}
  signed power transform
  \sep ARCH-type processes
  \sep GARCH
  \sep white noise tests
  \sep heteroskedasticity
  \sep sample autocorrelations
\end{keyword}

\end{frontmatter}

\section{Introduction}
\label{sec:introduction}



The signed power function, $\spower{x}{\lambda}$, has some useful properties which make it
suitable as a tool for statistical inference in time series analysis.  Indeed, unlike the
unisigned power function, the signed power is one-to-one.  Moreover, if the power is smaller
than one, then the transformed process has higher order moments than the original process,
unless the original process has infinitely many moments.  This makes it possible to use
inference results for the transformed process that would not be applicable to the
untransformed one.  For example, many results in time series require (or a much simpler) if
fourth order moments exist.  If a researcher is prepared to assume little more than, say, the
existence of second moments, then taking a signed square root (or smaller power) of the time
series will give a process for which such results can be applied.

The signed power has also a multiplicative property when one of the arguments is restricted
to be positive, Section~\ref{sec:sign-power-transf}, which is particularly useful when
dealing with ARCH-type processes.

We show that signed power transforms of some ARCH-type processes give ARCH-type processes
and, as a consequence, the transformed processes are white noise.  The class of ARCH-type
models for which this property holds contains many common ARCH and GARCH models.  The results
can be useful for computing significance bands for testing autocorrelations for equality to
zero and testing for white noise. Testing that the signed power transforms of a time series
are ARCH-type for different powers can give further confidence in the use of ARCH-type
models, since the signed powers of other types of white noise do not necessarilly remain
white noise.

\section{ARCH-type processes}
\label{sec:arch-type-processes}

We consider a general classs of heteroskedastic models which includes many ARCH/GARCH
specifications. The following definition is from \cite{kokoszka2011nonlinearity} \citep[see
also][pp. 10--11]{francq2011garchbook}.  As usual, $\tsmainfield{t}$ and
$\tsmainfieldX{t}{X}$ denote the $\sigma$-field generated by the process up to time~$t$,
where the upper index designates the process in case of ambiguity.

\begin{definition}  \label{def:archtype}

  A time series $\Ptsmain{t}$ is said to be an ARCH-type process if it obeys the
  following model:
  \begin{align} \label{eq:archtype}
    \tsmain{t} = \tsvolat{t}\tseta{t}
    , \qquad \text{where $\ptseta{t} \sim \IID(0,1)$}
    ,
  \end{align}
  together with the conditions:
  \begin{enumerate}
  \item[\emph{(i)}]

    For each $t$, $\tsvolat{t}$ is $\tsmainfield{t-1}$ measurable and square
    integrable.

  \item[\emph{(ii)}] 

    $\tsvolat{t} \ne \text{constant}$, for at least one $t$ in the domain.

  \item[\emph{(iii)}]

    For each $t$, $\tseta{t}$ is independent of $\tsmainfield{t-1}$, and it is
    square integrable with $\Ex \tseta{t} = 0$ and $\Ex \tseta{t}^2 > 0$.
\end{enumerate}

\end{definition}

ARCH-type models include many ARCH and GARCH specifications. They are obtained by specifying
the dynamics of $\tsvolat{t}$ by a separate equation. 


\section{Signed power  transformation}
\label{sec:sign-power-transf}

We define the signed power by
\begin{equation*}
  \brspower{x}{\lambda} = \sign(x)\abs{x}^{\lambda}
  .
\end{equation*}
This notation has been used before by \citet[Section 2.7]{samorodnitsky1994stable}.
\cite{yeo2000power} point out that as a transformation to normality the signed power has
certain defficiencies but we are not using it for this purpose, so this is not a concern.

The signed power function is monotone and one-to-one.  It has the useful multiplicative
property $\brspower{(xy)}{\lambda} = \brspower{x}{\lambda}\brspower{y}{\lambda}$ for
$x,y \in \mathbb{R}$.  In application to ARCH models it is sufficient for this property to
hold for $x > 0, y \in \mathbb{R}$.  A family of asymmetric power transforms having this
property can be defined by setting $f_{c}(x) = \abs{x}^{\lambda}$, $0$, or
$c \abs{x}^{\lambda}$ for $x > 0$, $x = 0$, and $x < 0$, respectively. Continuity is ensured
if $\lambda > 0$, see Appendix~\ref{sec:signed-power-as} for further details.  The signed
power transform is obtained by setting $c = -1$. Asymmetric signed power transforms obtained
with $c < 0$ are one-to-one and can be useful, for example, to ensure that the transformed
variable has mean zero, although we do not pursue this further here.

\section{Transformed ARCH-type processes}
\label{sec:transf-arch-type}

Let $\tssp{t} = \brspower{\tsmain{t}}{\lambda}$ and
$\tsspeta{t} = \brspower{\tseta{t}}{\lambda}$. Noting that $\tsvolat{t} > 0$ we get
\begin{align*}
  \tssp{t}
      &= \spower{\tsmain{t}}{\lambda}
       = \spower{\tsvolat{t}\tseta{t}}{\lambda}
       = \tsvolat{t}^{\lambda}\spower{\tseta{t}}{\lambda} 
       = \tsvolat{t}^{\lambda} \brspower{\tseta{t}}{\lambda}
   \\ &= \tsvolat{t}^{\lambda} \tsspeta{t}
  .
\end{align*}

\begin{theorem}
  \label{transf-archtype}
  Let ${\ptsmain{t}}$ be an ARCH-type process such that the distribution of $\tseta{t}$ is
  symmetric.
  Then for any $\lambda\in(0,1]$, the process $\tssp{t} = \brspower{\tsmain{t}}{\lambda}$ is
  also an ARCH-type process.
\end{theorem}
\begin{proof}
  The process $\ptsspeta{t}$ is i.i.d with $\Ex \tsspeta{t}^{2} < \infty$, since $\Ptseta{t}$
  has these properties.  Since the distribution of $\tseta{t}$ is symmetric, so is the
  distribution of $\tsspeta{t}$. Hence $\Ex \tsspeta{t} = 0$. Further, if $\tsvolat{t}$ is
  not a constant for at least one $t$, then so is $\tsvolat{t}^{\lambda}$.  The signed power
  transform is one-to-one, so the $\sigma$-fields generated by $\ptsmain{t}$ and $\ptssp{t}$
  are the same, i.e., $\tsmainfieldX{t-1}{Y} = \tsmainfieldX{t-1}{X}$.
  $\tsvolat{t}^{\lambda}$ is $\tsmainfieldX{t-1}{X}$ measurable since so is $\tsvolat{t}$ and
  $\tsvolat{t} > 0$. Hence, $\tsvolat{t}^{\lambda}$ is $\tsmainfieldX{t-1}{Y}$ measurable.
  $\tsvolat{t}^{\lambda}$ is square integrable for $\lambda \in (0,1]$, since $\tsvolat{t}$
  is.  Also, $\Ex \tsspeta{t}^{2}$ is finite and positive. So, conditions (i)--(iii) of
  Definition~\ref{def:archtype} hold for $\tsspeta{t}$ and so $\Ptssp{t}$ is an ARCH-type
  process.
\end{proof}

\begin{remark}
  The distribution of $\ptsmain{t}$ in Theorem~\ref{transf-archtype} is not necessarilly
  symmetric.  The symmetry requirement is imposed only on $\tseta{t}$. This is to ensure that
  the mean remains zero after the transformation. If $\tseta{t}$ is restricted to a specific
  family of distributions, it may be possible to relax the symmetry requirement with the help
  of the asymetric power transform discussed in Section~\ref{sec:sign-power-transf} by using
  suitable $c = c(\lambda)$.
\end{remark}

\begin{remark}
  The requirement that $\lambda \in (0,1]$ is sufficient for our purposes and simplifies
  somewhat the statement of the result. $\lambda$ need only be small enough for the second
  moments to exist.
\end{remark}

\begin{remark}
  The moment conditions in Theorem~\ref{transf-archtype} are on $\tseta{t}$, not
  $\tsmain{t}$. Existence of moments of certain order for $\tseta{t}$ does not guarantee the
  same for $\tsmain{t}$. The latter may not have even second order moments.
\end{remark}


Asymptotic theory is significantly simpler for processes with finite fourth moments.
The transformed process will have this property at least for sufficiently small powers.
We formulate the result in the following theorem.

\begin{theorem} \label{th:sp-is-arch}
  Let ${\ptsmain{t}}$ be an ARCH-type process, such that the distribution of $\tseta{t}$ is
  symmetric (around zero). Let
  $\alpha^{*} = \sup \{\alpha : \Ex{\abs{\tsmain{t}}^{\alpha}}<\infty\}$. Assume that
  $\alpha^{*}>1$.
  Then the process $\tssp{t} = \brspower{\tsmain{t}}{\lambda}$ is an ARCH-type process with
  finite fourth moment for any $\lambda\in(0,\alpha^{*}/4)$, hence at least for
  $\lambda\in(0,1/4)$.
\end{theorem}
\begin{proof}
  The result follows from Theorem~\ref{transf-archtype} by observing that if
  $\lambda\in(0,\alpha^*/4)$ then the fourth moments exist.
\end{proof}

\begin{remark}
  It is common to assume that the second moment of $\tsmain{t}$ exists. Then $\lambda^{*}$ in
  Theorem~\ref{th:sp-is-arch} is at least $2$. So, in this case the signed power
  transformation gives a process with finite fourth moment at least for
  $\lambda \in (0,2/4) \equiv (0,1/2)$.
\end{remark}

\begin{remark} \label{remark:bracket}
  If $\Ex{\abs{\tsmain{t}}^{\alpha^{*}}}<\infty$ then the fourth moment exists for $\lambda$
  in the semi-closed interval $\lambda\in(0,\alpha^*/4]$.
\end{remark}


Since the signed power transform is one-to-one, the $\sigma$-fields generated by $\Ptssp{t}$
are the same as for the original process, $\Ptsmain{t}$. Hence $\Ptssp{t}$ inherits all
properties that depend only on these $\sigma$-fields, most notably mixing properties. We
have, for example, the following result.


\begin{theorem} \label{th:mixing}
  Under the conditions of Theorem~\ref{th:sp-is-arch}, if in addition $\Ptsmain{t}$ has 
  exponentilly decaying $\alpha$-mixing coefficients, then the same holds for $\Ptssp{t}$.
\end{theorem}

The mixing condition in Theorem~\ref{th:mixing} is not a problem, since for most GARCH models
the stronger $\beta$-mixing condition with exponential rates holds
\citetext{\citealp{CarrascoChen2002}; \citealp[Chapter~3]{francq2011garchbook}}.

\section{Correlogram of transformed series}
\label{sec:corr-transf-seirs}

%
%
Let $\Ptsmain{t}$ be a second order stationary process,
$\acvf{\ell} = \Cov(\tsmain{t},\tsmain{t-\ell})$, $\acf{\ell} = \acvf{\ell} /
\acvf{0}$. 
Let also
$ \bar{X}_n = \frac{1}{n}\sum_{k=1}^n \tsmain{k} $,            
$\acvfhat{\ell}
  = \frac{1}{n} \sum_{k=1}^{n-\ell} \left(X_k -\bar{X}_n \right) 
                                  \left(X_{k+\ell}-\bar{X}_n \right)
$,
$ \acfhat{\ell}  = \acvfhat{\ell}/\acvfhat{0} $
be the sample mean, sample autocovariances and sample autocorrelations for a stretch of
length~$n$. 

We are interested in cases when $\boldacfhat = \left[\acf{1},\acf{2},\dots,\acf{m} \right]^T$
is asymptotically normal, i.e.
\begin{align} \label{eq:asynormal}
  \sqrt{n} (\boldacfhat - \boldacf)
  \stackrel{\mathrm{d}}{\to} 
  \mathcal{N}(\textbf{0}, \textbf{W})
  \qquad \text{as $n\to \infty$}
  ,
\end{align}
where $\boldacf = \left[\acf{1},\acf{2},\dots,\acf{m} \right]^T$.

%
$\ptsmain{t}$ is said to be linear if
$\tsmain{t} = \sum_{-\infty}^{\infty} \psi_{l}\varepsilon_{t-l}$, where $\varepsilon_{t}$ is
i.i.d. and $\sum \abs{\psi_{l}} < \infty$.  If $\ptsmain{t}$ is linear and $\varepsilon_{t}$
has finite fourth moment, then asymptotic normality holds and the entries of the $m \times m$
matrix $\textbf{W}$ are given by the Bartlett's formula
\citep[Theorem~1.1]{francq2011garchbook}
\begin{align*}
  w_{ij} 
  = \sum_{k=-\infty}^{\infty} 
  \left(\acf{k+i} \acf{k+j} + \acf{k-i}\acf{k+j} + 2\acf{i} \acf{k} \acf{k}^2
        - 2\acf{i} \acf{k} \acf{k+j} - 2\acf{j} \acf{k} \acf{k+i} \right)
  .
\end{align*}
If fourth moments exist but the process is not linear, asymptotic normality still holds but
the asymptotic distribution is not universal and depends on further properties of the
process. Results for processes that are not linear have been given by
\cite{RomanoThombs1996}, \cite{FrancqRoyZakoian2005}, and others. A general result,
containing the Bartlett's formula as a particular case is given by \citet[Theorem~B5]{francq2011garchbook}.

For linear processes the fourth moment requirement on $\varepsilon_{t}$ can be relaxed to
just $\Ex \varepsilon_{t}^{2} < \infty$ if $\sum_{l}\abs{l}\psi _{l}^{2} < \infty$.
However, this is a very specific case. In general, when
the largest moment of a process (not necessarilly linear) is in the interval (2,4), the
limiting distributions of the sample autocorrelations are stable with index smaller than~2
and normalising constants vastly different from~$\sqrt{n}$ \citep{DavisMikosch1998}.

For routine work it is therefore convenient to work under assumptions that ensure asymptotic
normality.
\citet[Proposition 3.1]{kokoszka2011nonlinearity} have shown that if an ARCH-type
process $\Ptsmain{t}$ is $\alpha$-mixing with exponentially decaying
coefficients and finite $(4+\delta)$th moment, then the sample autocorrelations
are assymptoticaly normal with diagonal covariance matrix with elements given by
\begin{align}
  \label{eq:wij}
  w_{ij}=\delta_{ij} \frac{E\left[X_1^2 X_{1+i}^2 \right]}{\left(EX_1^2 \right)^2}
  ,
\end{align}
where $\delta_{ij}$ is the Kronecker delta. Further, $w_{ii}$ can be estimated
consistently by 
\begin{equation}
  \label{eq:3}
  \hat{w}_{ij} = nc_{in} \frac{\sum_{d=1}^{n-i}X_{d}^{2}X_{d+i}^{2}}{
                            \left( \sum_{d=1}^{n}X_{d}^{2} \right)^{2} }
   , 
\end{equation}
where $c_{in}$ are asymptotically equal to one, e.g. $c_{in} = n/(n-i)$ or
$c_{in} = 1$ \citep[Proposition 3.3]{kokoszka2011nonlinearity}.

When $\delta$ is close to zero, the variance is large. So, if the process does not have 4th
moments or has moments just over four, inference based on this result may not be reliable.

The signed power transform can be used to ensure more reliable inference by transforming the
process to one with higher finite moments. From the results in the previous section and
\citet[Proposition 3.1]{kokoszka2011nonlinearity} we obtain the following result.

\begin{theorem}
  \label{spower-acf}
  Suppose that $\Ptsmain{t}$ is an ARCH-type process with exponentially decaying
  $\alpha$-mixing coefficients and finite $s\ge2$ moment. Let
  $\tssp{t} = \brspower{\tsmain{t}}{\lambda}$ where $\lambda\in(0,s/4)$.  Let also
  $\acfhat{i}$, $i=1,\dots,m$, be the sample autocorrelations of the transformed time series,
  $\tssp{t}$, ${w}_{ij}$ and $\hat{w}_{ii}$ be computed as in
  Equations~\eqref{eq:wij}--\eqref{eq:3} with $X_{t}$ replaced by $\tssp{t}$.

  Then Equation~\eqref{eq:asynormal} holds with elements of the diagonal matrix $W$ given
  by $w_{ij}$. The diagonal elements of $W$ can be estimated by $\hat{w}_{ii}$.
\end{theorem}

Portmanteau tests can also be carried out.

\begin{corollary}[Portmanteau test]
  Under the conditions of Theorem~\ref{spower-acf} we have
  \begin{align*}
    n \sum_{i=1}^{m} \frac{\acfhat{i}^{2}}{\hat{w}_{ii}}
    \sim \chi^{2}_{m}
    \qquad \text{(approximately)}
    .
  \end{align*}
\end{corollary}

There are various ways to use this result.  If it is suspected that the highest moment of a
time series is, for example, in (2,4), then significance bounds can be constructed, based on
the above result, for the sample autocorrelations of the signed power transformed series with
one or more values for $\lambda \le 1/2$, say $\lambda=1/3$. The same significance bounds
remain valid if the time series has finite moments larger than~4, so not much (if anything)
is lost by being cautious.  The same applies to the portmanteau test.

It may also be prudent to verify that signed powers are white noise even if the moments of
the time series are not of concern. Moreover, Theorem~\ref{th:sp-is-arch} shows that although
an ARCH-type process is weak white noise, it is not so weak.  The signed powers of other
types of white noise do not need to be white noise. So these tools can be used to distinguish
ARCH-type white noise from other types of white noise by considering a range of $\lambda$'s.

\section{Examples}
\label{sec:illustr-exampl}

In this section we consider some simulation examples to illustrate the ideas.
The computations were done with R \citep{Rbase}.

\subsection{ARCH(1) model}
\label{sec:arch1-model}

Let $\ptsmain{t}$ be an ARCH(1) process specified by
\begin{align} \label{eq:arch1}
  \tsmain{t} &= \tsvolat{t} \tseta{t} ,\\
  \tsvolat{t}^{2} &= 0.01 + \alpha_{1} \tsmain{t}^{2} \nonumber
                   ,                 
\end{align}
where $\ptseta{t}$ is i.i.d., $N(0,1)$, and independent of $\tsmainfield{t-1}$. 

In this example we explore the performance of the statistics for values of the parameter
$\alpha_{1} = 0.05, 0.15, \dots, 0.95$. This model was considered by
\citet{kokoszka2011nonlinearity} who put emphasis on the quality of estimation of $w_{11}$,
while we look at its effect on the rejection rate of the underlying null hypothesis. We
include also values of $\alpha_{1} > 1/\sqrt{3} \approx 0.577$, for which the process doesn't
have fourth moments and so signed power transformation is needed to ensure validity of the
theory.

For each $\alpha_{1}$, we simulated 10,000 time series of length $n=2000$ from the model.
For each time series, $\ptsmain{t}$, we calculated $\acfhat{1}$ and $\hat{w}_{11}$ of the
time series and its signed powers $\lambda = 0.1, 0.2, \dots, 0.9$.
Table~\ref{tab:arch1:n2000} shows the fraction of times when the lag~1 sample autocorrelation
coefficient was outside the 95\% limits suggested by Theorem~\ref{th:sp-is-arch} and
Remark~\ref{remark:bracket}.  Each row gives the results for one value of $\alpha_{1}$. The
second column gives the theoretical values of $w_{11}$ for the (untransformed) ARCH(1)
process. The last four elements of this column contain NA's since the process doesn't possess
fourth moments for $\alpha_{1} > 1/\sqrt{3}$.
\begin{table}

\caption{\label{tab:arch1:n2000}ARCH(1) example, $n = 2000$.}
\centering
\begin{tabular}[t]{rrrrrr}
\toprule
$\beta_{1}$ & $w_{ii}$ & $\lambda = 0.1$ & $\lambda = 0.5$ & $\lambda = 0.75$ & $\lambda = 1$\\
\midrule
0.05 & 1.101 & 0.049 & 0.048 & 0.049 & 0.051\\
0.15 & 1.322 & 0.051 & 0.053 & 0.048 & 0.053\\
0.25 & 1.615 & 0.051 & 0.053 & 0.048 & 0.051\\
0.35 & 2.107 & 0.052 & 0.049 & 0.051 & 0.051\\
0.45 & 3.293 & 0.051 & 0.046 & 0.046 & 0.045\\
\addlinespace
0.55 & 12.892 & 0.051 & 0.052 & 0.052 & 0.047\\
0.65 & NA & 0.049 & 0.049 & 0.052 & 0.043\\
0.75 & NA & 0.050 & 0.048 & 0.048 & 0.045\\
0.85 & NA & 0.052 & 0.050 & 0.045 & 0.042\\
0.95 & NA & 0.053 & 0.049 & 0.046 & 0.039\\
\bottomrule
\end{tabular}
\end{table}
The last column, $\lambda = 1$, is for the untransformed time series.

We calculated $\hat{w}_{11}$ even for $\alpha_{1}$ when the asymptotic distribution is not
valid. For $\alpha_{1}$ up to  $0.35$ the empirical rate is very close to the nominal
5\%. For the rest it decreases to $0.039$ as the value of the parameter approaches the
boundary $\alpha_{1} = 1$ of the region where second moments exist. Note that $\alpha_{1}=
0.55$ is close to the fourth moment border, $1/\sqrt{3}$, so $w_{11}$ jumps to $12.892$ but
the rejection rate doesn't deteriorate catastrophically.

According to Theorem~\ref{th:sp-is-arch} the test is guaranteed to work for
$\brspower{\tsmain{t}}{\lambda}$ with $\lambda \le 0.5$ for all values of $\alpha_{1}$ in
$(0,1)$. The columns for $\lambda = 0.5$ and $\lambda = 0.1$ are indeed around $0.5$ for all
values of the parameter. The same is true for $\lambda = 0.75$ even though
$\brspower{\tsmain{t}}{0.75}$ may not have fourth moment for the largest $\alpha_{1}$'s.

To check the performance with shorter time series, we did also a similar simulation with
shorter time series ($n = 500$). The results were very similar and not included here.

In summary, this example gives results expected in view of Theorem~\ref{th:sp-is-arch} and
supporting the theory for the case when the underlying process is ARCH-type. Decisions based
on $\brspower{\tsmain{t}}{\lambda}$ are somewhat more likely to be correct, especially when
the fourth moments don't exist, but the main benefit of the signed transform here is that of
giving reassurance that an ARCH-type model is suitable.


\subsection{Non-ARCH-type white noise}
\label{sec:non-arch-type}

This example illustrates the usefulness of the signed power transform for detecting that
the underlying process is not ARCH-type even when it is white noise.

Mixture autoregressive (MAR) models \citep{WongLi2000} are in general autocorrelated but they
are white noise on a subset of the parameter space
\citep{ravagli2020bayesian}.
We consider here the following MAR(2;1,1) models 
\begin{align}                                               \nonumber
  F(X_{t} |\tsmainfield{t-1})
  &\equiv \Pr(X_{t}\le x | \tsmainfield{t-1}) 
  \\ &= 0.25 F_{0}(X_{t} - 0.3X_{t-1})
       + 0.75 F_{0}((X_{t} + 0.1X_{t-1}) / \sigma_{2})        \label{eq:marwn} 
       ,
\end{align}
where $F_{0}$ is the distribution function of the standardised~$t_{3}$ distribution and
$\sigma_{2}$ is a scale parameter. These models are white noise but not ARCH-type and we are
interested if the proposed methodology can help to detect this. Also,
$\Ex |\tsmain{t}|^{a} =\infty$ for $a \ge 3$.

We did a small simulation study using the model in Equation~\eqref{eq:marwn} with
$\sigma_2 = 1,\dots,10$.  For each $\sigma_{2}$, we simulated 10,000 time series of length
$n=2000$ from the model \citep[using R package mixAR, see][]{boshnakov2020mixar}.  For each time series, $\ptsmain{t}$, we calculated $\acfhat{1}$ and
$\hat{w}_{11}$ of the time series and its signed powers $\lambda = 0.1, 0.2, \dots, 0.9$.
%
%
Table~\ref{tab:t3:n2000} shows the fraction of times when the lag~1 sample autocorrelation
coefficient was outside the 95\% limits. suggested by Theorem~\ref{th:sp-is-arch}.  The
$(i,j)$th cell in the table gives the result for $\sigma_{2}= i$ and $\lambda = j/10$ (so,
the last column is for $\lambda = 1$, the untransformed time series).  We expect these values
to be around 5\% when the null hypothesis that the time series are ARCH-type is correct. We
hope also that the test will be able to detect departures from $H_{0}$, i.e., it will reject
$H_{0}$ more often when the underlying time series are not ARCH-type.
%
%

\begin{table}

\caption{\label{tab:t3:n2000}Empirical p-values for MAR model with $F_0(\cdot)$ cdf of standardised-$t_3$ distribution. The first column is $\sigma_{2}$. Each column gives the results for a given $\lambda$. $n = 2000$.}
\centering
\begin{tabular}[t]{lrrrrrrrrrr}
\toprule
  & 0.1 & 0.2 & 0.3 & 0.4 & 0.5 & 0.6 & 0.7 & 0.8 & 0.9 & 1\\
\midrule
1 & 0.056 & 0.061 & 0.058 & 0.056 & 0.057 & 0.056 & 0.055 & 0.054 & 0.048 & 0.044\\
2 & 0.269 & 0.227 & 0.173 & 0.130 & 0.100 & 0.071 & 0.055 & 0.048 & 0.047 & 0.044\\
3 & 0.632 & 0.547 & 0.416 & 0.292 & 0.195 & 0.118 & 0.072 & 0.053 & 0.044 & 0.042\\
4 & 0.855 & 0.763 & 0.618 & 0.438 & 0.276 & 0.156 & 0.091 & 0.058 & 0.047 & 0.043\\
5 & 0.945 & 0.877 & 0.740 & 0.539 & 0.340 & 0.186 & 0.096 & 0.061 & 0.046 & 0.042\\
\addlinespace
6 & 0.977 & 0.933 & 0.807 & 0.608 & 0.383 & 0.214 & 0.102 & 0.063 & 0.045 & 0.045\\
7 & 0.990 & 0.958 & 0.863 & 0.662 & 0.423 & 0.218 & 0.110 & 0.063 & 0.047 & 0.043\\
8 & 0.995 & 0.974 & 0.895 & 0.697 & 0.444 & 0.236 & 0.112 & 0.062 & 0.048 & 0.043\\
9 & 0.998 & 0.983 & 0.917 & 0.736 & 0.461 & 0.245 & 0.116 & 0.063 & 0.043 & 0.044\\
10 & 0.998 & 0.987 & 0.928 & 0.749 & 0.488 & 0.247 & 0.125 & 0.066 & 0.046 & 0.045\\
\bottomrule
\end{tabular}
\end{table}

We see from Table~\ref{tab:t3:n2000} that the values for the untransformed time series (last
column) are between 4.2\%--4.5\%, close but just over half percentage point below the nominal
5\%.  As $\lambda$ decreases the rejection rate increases, i.e., we are more likely to reject
the ARCH-type hypothesis. For example, for $\lambda = 0.5$ and $\sigma_{2}= 7$ we reject the
ARCH-type white noise hypothesis in about 40\% of the cases. With $\lambda = 0.1$ the
rejection rates are mostly over 90\%. An exception is the model with $\sigma_{2}=1$ for which
all values stay close to 5\%.

To check how the performance with shorter time series, we did also a similar simulation with
shorter time series ($n = 500$)%
. The results were qualitatively similar but
more conservative.

%

In summary, the signed power transform in this example allows as to detect that the
underlying process is white noise but not of ARCH-type.

%
%
% ...
%
%
%
% ... 

\bibliographystyle{elsarticle-harv}
\bibliography{signedpower}

\begin{thebibliography}{12}
\expandafter\ifx\csname natexlab\endcsname\relax\def\natexlab#1{#1}\fi
\providecommand{\url}[1]{\texttt{#1}}
\providecommand{\href}[2]{#2}
\providecommand{\path}[1]{#1}
\providecommand{\DOIprefix}{doi:}
\providecommand{\ArXivprefix}{arXiv:}
\providecommand{\URLprefix}{URL: }
\providecommand{\Pubmedprefix}{pmid:}
\providecommand{\doi}[1]{\href{http://dx.doi.org/#1}{\path{#1}}}
\providecommand{\Pubmed}[1]{\href{pmid:#1}{\path{#1}}}
\providecommand{\bibinfo}[2]{#2}
\ifx\xfnm\relax \def\xfnm[#1]{\unskip,\space#1}\fi
\bibitem[{Boshnakov and Ravagli(2020)}]{boshnakov2020mixar}
\bibinfo{author}{Boshnakov, G.N.}, \bibinfo{author}{Ravagli, D.},
  \bibinfo{year}{2020}.
\newblock \bibinfo{title}{mix{AR}: mixture autoregressive models, {R} package
  version 0.22.4}.
\newblock \bibinfo{howpublished}{https://CRAN.R-project.org/package=mixAR}.
\bibitem[{Carrasco and Chen(2002)}]{CarrascoChen2002}
\bibinfo{author}{Carrasco, M.}, \bibinfo{author}{Chen, X.},
  \bibinfo{year}{2002}.
\newblock \bibinfo{title}{Mixing and moment properties of various {GARCH} and
  stochastic volatility models}.
\newblock \bibinfo{journal}{Econometric Theory} \bibinfo{volume}{18},
  \bibinfo{pages}{17--39}.
\newblock \URLprefix \url{http://www.jstor.org/stable/3533024},
  \DOIprefix\doi{10.2307/3533024}.
\bibitem[{Davis and Mikosch(1998)}]{DavisMikosch1998}
\bibinfo{author}{Davis, R.A.}, \bibinfo{author}{Mikosch, T.},
  \bibinfo{year}{1998}.
\newblock \bibinfo{title}{The sample autocorrelations of heavy-tailed processes
  with applications to {ARCH}}.
\newblock \bibinfo{journal}{The Annals of Statistics} \bibinfo{volume}{26},
  \bibinfo{pages}{2049--2080}.
\newblock \URLprefix \url{http://www.jstor.org/stable/120032}.
\bibitem[{Francq et~al.(2005)Francq, Roy and Zakoïan}]{FrancqRoyZakoian2005}
\bibinfo{author}{Francq, C.}, \bibinfo{author}{Roy, R.},
  \bibinfo{author}{Zakoïan, J.M.}, \bibinfo{year}{2005}.
\newblock \bibinfo{title}{Diagnostic checking in arma models with uncorrelated
  errors}.
\newblock \bibinfo{journal}{Journal of the American Statistical Association}
  \bibinfo{volume}{100}, \bibinfo{pages}{532--544}.
\newblock \URLprefix \url{https://doi.org/10.1198/016214504000001510},
  \DOIprefix\doi{10.1198/016214504000001510}.
\bibitem[{Francq and Zakoian(2011)}]{francq2011garchbook}
\bibinfo{author}{Francq, C.}, \bibinfo{author}{Zakoian, J.M.},
  \bibinfo{year}{2011}.
\newblock \bibinfo{title}{{GARCH} models: structure, statistical inference and
  financial applications}.
\newblock \bibinfo{publisher}{John Wiley \& Sons}.
\bibitem[{Kokoszka and Politis(2011)}]{kokoszka2011nonlinearity}
\bibinfo{author}{Kokoszka, P.}, \bibinfo{author}{Politis, D.},
  \bibinfo{year}{2011}.
\newblock \bibinfo{title}{Nonlinearity of {ARCH} and stochastic volatility
  models and bartlett's formula}.
\newblock \bibinfo{journal}{Probability and Mathematical Statistics}
  \bibinfo{volume}{31}, \bibinfo{pages}{47--59}.
\bibitem[{{R Core Team}(2021)}]{Rbase}
\bibinfo{author}{{R Core Team}}, \bibinfo{year}{2021}.
\newblock \bibinfo{title}{R: A Language and Environment for Statistical
  Computing}.
\newblock \bibinfo{organization}{R Foundation for Statistical Computing}.
  \bibinfo{address}{Vienna, Austria}.
\newblock \URLprefix \url{https://www.R-project.org/}.
\bibitem[{Ravagli and Boshnakov(2020)}]{ravagli2020bayesian}
\bibinfo{author}{Ravagli, D.}, \bibinfo{author}{Boshnakov, G.N.},
  \bibinfo{year}{2020}.
\newblock \bibinfo{title}{Bayesian analysis of mixture autoregressive models
  covering the complete parameter space}.
\newblock \URLprefix \url{https://arxiv.org/abs/2006.11041},
  \href{http://arxiv.org/abs/2006.11041}{{\tt arXiv:2006.11041}}.
\bibitem[{Romano and Thombs(1996)}]{RomanoThombs1996}
\bibinfo{author}{Romano, J.P.}, \bibinfo{author}{Thombs, L.A.},
  \bibinfo{year}{1996}.
\newblock \bibinfo{title}{Inference for autocorrelations under weak
  assumptions}.
\newblock \bibinfo{journal}{Journal of the American Statistical Association}
  \bibinfo{volume}{91}, \bibinfo{pages}{590--600}.
\newblock \DOIprefix\doi{10.1080/01621459.1996.10476928},
  \href{http://arxiv.org/abs/https://www.tandfonline.com/doi/pdf/10.1080/01621459.1996.10476928}{{\tt
  arXiv:https://www.tandfonline.com/doi/pdf/10.1080/01621459.1996.10476928}}.
\bibitem[{Samorodnitsky and Taqqu(1994)}]{samorodnitsky1994stable}
\bibinfo{author}{Samorodnitsky, G.}, \bibinfo{author}{Taqqu, M.S.},
  \bibinfo{year}{1994}.
\newblock \bibinfo{title}{Stable non-Gaussian random processes: stochastic
  models with infinite variance}. volume~\bibinfo{volume}{1}.
\newblock \bibinfo{publisher}{CRC press}.
\bibitem[{Wong and Li(2000)}]{WongLi2000}
\bibinfo{author}{Wong, C.S.}, \bibinfo{author}{Li, W.K.}, \bibinfo{year}{2000}.
\newblock \bibinfo{title}{{On a mixture autoregressive model.}}
\newblock \bibinfo{journal}{J. R. Stat. Soc., Ser. B, Stat. Methodol.}
  \bibinfo{volume}{62}, \bibinfo{pages}{95--115}.
\bibitem[{Yeo and Johnson(2000)}]{yeo2000power}
\bibinfo{author}{Yeo, I.K.}, \bibinfo{author}{Johnson, R.A.},
  \bibinfo{year}{2000}.
\newblock \bibinfo{title}{A new family of power transformations to improve
  normality or symmetry}.
\newblock \bibinfo{journal}{Biometrika} \bibinfo{volume}{87},
  \bibinfo{pages}{954--959}.

\end{thebibliography}

\appendix

\section{Signed power as  a solution to a functional equation}
\label{sec:signed-power-as}


\begin{lemma}
  Consider the functional equation 
  \begin{align}
    \label{eq:funxy}
    f(xy) = f(x)f(y)
    , \qquad
    \{ (x,y) : [0,\infty) \times \mathbb{R} \}
    ,
  \end{align}
  where $f: \mathbb{R} \to \mathbb{R}$ is continuous for $x \neq 0$.  The solutions of
  Equation~\eqref{eq:funxy} are:
  \begin{align}
    f_{0}(x) &\equiv 0   
               , \qquad
    f_{1}(x) \equiv 1   
               , \qquad
    f(x) =            
    \begin{cases}
      \abs{x}^{\lambda}   &\text{for $x > 0$,} \\
      0             & \text{for $x = 0$,} \\
      c \abs{x}^{\lambda} &\text{for $x<0$,} 
    \end{cases}  \label{eq:f2}
  \end{align}
  where $c$ is arbitrary. $f(x)$ is continuous on $\mathbf{R}$ when $a > 0$, otherwise it is
  continuous on $\mathbf{R}\setminus\{0\}$.

\end{lemma}

Notice the restriction, $x\ge0$, imposed on $x$ but not $y$. This is because we need this
property for $\brspower{(\sigma_{t}\eta_{t})}{\lambda}$, where $\sigma_{t}\ge0$. 


\begin{proof}
  
  That $f_{0}()$, $f_{1}()$, and $f()$, given by Equation~\eqref{eq:f2} are solutions can be
  verified directly.

Further, $f_{1}()$ is the only solution (continuous or not) with $f(0) \neq 0$. Indeed,
setting $x = y = 0$ in Equation~\eqref{eq:funxy} we get $f(0) = f(0)^{2}$, so $f(0) = 1$ or
$f(0) = 0$. Similarly, for $x = 0$ and arbitrary $y$ we have $f(0) = f(0)f(y)$. The last
equation implies that $f(0) = 1$ gives a single solution $f(y) \equiv 1$.  So, $f(0) = 0$ for
any other solution.

Similarly, $f_{0}()$ is the only solution with $f(y) = 0$ for some $y > 0$ and
$f(1) \neq 1$. Indeed, setting $x=y=1$ in Equation~\eqref{eq:funxy} gives $f(1) = f(1)^{2}$,
so $f(1) = 0$ or $f(1) = 1$.  But $f(1) = 0$ implies that $f(y) = 0$ for all $y$ (which is
$f_{0}$), so for any other solution $f(1) = 1$. Also, if $f(y) = 0$ for some $y > 0$, then
for any $z$ we have $f(z) = f(y) f(z/y) = 0$ (i.e., $f_{0}(z)$).
Note that this last argument doesn't apply for $y < 0$ and $z > 0$, since then both $y$ and
$z/y$ will be negative, in which case  Equation~\eqref{eq:funxy} is not required to hold.

Consider $y<0$. Write $y = \abs{y}(-1)$ and use Equation~\eqref{eq:funxy} to get
\begin{align} \label{eq:fext}
  f(y) &= f(\abs{y}(-1)) = f(\abs{y})f(-1)
  .
\end{align}
So, for $y<0$ we have $f(y) = c f(\abs{y})$, where $c = f(-1)$.  Now, let $f(y)$ be any
solution of \eqref{eq:funxy} restricted to the non-negative real numbers, i.e.,
$f(xy) = f(x)f(y)$ for $x,y\ge0$.  The above calculations show that we can extend it to the
whole real axis by setting $f(y) = c f(\abs{y})$ for $y < 0$, where $c$ can be any constant
when $f(0) = 0$, otherwise (i.e., if $f(0)=1$) $c = 1$, giving $f_{1}()$, as discussed above.


If $x$ and $y$ are both restricted to have non-negative values, it is well known that the
non-trivial continuous (for $y > 0$) solutions of $f(xy) = f(x)f(y)$ are $f(y)=y^{\lambda}$ for
$y>0$ and $f(0)=0$, where $a\neq0$. Another solution is $f(x) = 1$ for $x>0$ and $f(0) = 0$
but it can be considered a special case of $y^{\lambda}$ with $a = 0$. Extending these solutions
to $y < 0$ using Equation~\eqref{eq:fext} we get Equation~\eqref{eq:f2}.

Clearly, there are no other solutions to Equation~\eqref{eq:funxy}.
\end{proof}

Choosing $c=-1$ gives signed power transform defined in the main text. In this case,
$\lambda = 0$ gives the sign function.

 \end{document}